\documentclass[onecolumn,journal]{IEEEtran}

\ifCLASSINFOpdf
\else
\fi

\renewcommand{\baselinestretch}{1.75}
\usepackage{amsmath,amssymb}
\usepackage{algpseudocode}
\usepackage{cite}
\usepackage[caption=false,font=footnotesize]{subfig}
\usepackage[pdftex]{graphicx}
\usepackage{epstopdf} 
\DeclareGraphicsExtensions{.jpg,pdf,.png}
\usepackage{amsthm}
\usepackage{algorithm}
\usepackage{multirow}
\newtheorem{theorem}{Theorem}

\theoremstyle{definition}
\newtheorem{definition}{Definition}

\theoremstyle{remark}
\newtheorem{remark}{Remark}

\date{}
\begin{document}
\title{Achievable Secrecy Rate for the Relay Eavesdropper Channel with Non-Causal State Available at Transmitter and Relay}
\author{Nematollah Zarmehi\thanks{N. Zarmehi is with the Advanced Communication Research Institute (ACRI), Electrical Engineering Department, Sharif University of Technology, Tehran, Iran. (e-mail: zarmehi\_n@ee.sharif.edu)}}
\markboth{}{N. Zarmehi: Achievable Secrecy Rate for the Relay Eavesdropper Channel with Non-Causal State Available at Transmitter and Relay}
\maketitle

\begin{abstract}
In this paper, we consider a more general four-terminal memoryless relay-eavesdropper channel with state information (REC-SI) and derive an achievable perfect secrecy rate for it. We suppose that the state information is non-causally available at the transmitter and relay only. The transmitter wishes to establish a secure communication with the legitimate receiver by the help of a relay where a confidential message will be kept secret from a passive eavesdropper. We consider active cooperation between the relay and transmitter. The relay helps the transmitter by relaying the message using decode-and-forward (DF) scheme. The proposed model is a generalization of some existing models and the derived achievable perfect secrecy rate is compared to the special cases. The results are also validated numerically for the additive white Gaussian noise (AWGN) channel.
\end{abstract}



\section{Introduction}\label{sec:intro}
\IEEEPARstart{S}{hannon} developed the theory of secrecy systems in \cite{ref:shannonSecSys}, in which the transmitter attempts to send message $M$ to a legitimate receiver through a noiseless channel. This message is encrypted to a cipher text $E$ by a key ($K$) shared between the transmitter and receiver. In Shannon's model, an eavesdropper which knows the family of encryption and decryption functions has access to the ciphertext $E$ and tries to get some information about the message $M$. In Shannon model \cite{ref:shannonSecSys}, the perfect secrecy can be achieved if $p(M|E)=P(M)$. In particular, the system is considered perfectly secure if the key space is at least as large as the message space.

The basic information theoretic model of wiretap channel was introduced by Wyner \cite{ref:wtWyner}. In this model, the eavesdropper as a passive receiver gets a noisy version of the legitimate receiver's signal. Then, this degraded wiretap channel was extended to the broadcast channel with confidential messages by Csisz\'ar and K\"orner in \cite{ref:csiszarKorner}, where the transmitter attempts to send a common message to both legitimate receiver and eavesdropper and to send a secret message only to the legitimate receiver.

Considering the works on communications with state information \cite{ref:shannonSideInf,ref:gelfandPinsker}, Chen and Vink investigated the Gaussian wiretap channel with state information \cite{ref:mitrpantGaussWtSide}. This work was then extended to the wiretap channel with state information in \cite{ref:chen2008} and after that, several similar ideas were investigated \cite{ref:elgamalCausalWt,ref:khistiKeyAgr,ref:BocheWtStrngSecrecy,ref:villardSecTrns,ref:zibaeenejad}. 

During the past years, a number of international efforts have led to developing several initiatives on the concept of cooperative secrecy. These works deal with three main categories, oblivious cooperation for secrecy, active cooperation for secrecy, and untrusted helpers \cite{ref:ekremCoopSec}. In the first category, the cooperating party facilitates the transmission of a confidential message from the transmitter to the receiver without the need of any information about that message. Cooperative jamming \cite{ref:coopJamming}, artificial noise \cite{ref:artfNoise}, and noise forwarding \cite{ref:elgamalREC} are three common techniques proposed for this kind of cooperation. In the second category, the cooperating party is active and facilitates the transmission by relaying the message using some techniques like decode and forward (DF) and compress and forward \cite{ref:relayWt,ref:elgamalREC}. Finally, in the third category, the eavesdropper is assumed to be a cooperating partner and the interaction of cooperation and secrecy is investigated \cite{ref:xiangUnHelper,ref:elgamalCapClassRc}.

In this paper, we consider relay eavesdropper channel with an active cooperation. Moreover, we assume that the channel is state-dependent and a state information is non-causally available at the transmitter and relay. Exploiting the results of cooperation for secrecy in relay eavesdropper channel \cite{ref:elgamalREC} and state-dependent relay channel \cite{ref:zaidiRcSide}, we consider a general framework and propose achievable secrecy rate for the memoryless relay eavesdropper channel with state information (REC-SI). In this model, the relay uses DF scheme to facilitate the transmission of the confidential message from the transmitter to the receiver. We extend this work to the Gaussian channel and provide some numerical examples.

The organization of the paper is as follows. Section \ref{sec:notations-sysModel} introduces our notations and the discrete memoryless REC-SI. An achievable secrecy rate for the REC-SI is presented in Section \ref{sec:achrate}. In Section \ref{sec:numex}, we extend the results to the Gaussian REC-SI. We provide some numerical examples of Gaussian REC-SI in Section \ref{sec:numex}. Finally, Section \ref{sec:con} concludes the paper.

\section{Notations and System Model}\label{sec:notations-sysModel}
\subsection{Notations}\label{subsec:notations}
We use the following notations throughout the paper. Uppercase letters are used to denote random variables, (e.g. $X$), and lowercase letters their realizations, (e.g. $x$). Also, calligraphic letters are used to denote the alphabet set, (e.g. $\mathcal{X}$). Vectors are denoted as $\mathbf{X^n}=\bigl({X(1),X(2),X(3),...,X(n)}\bigr)$. We use the following shorthand for probability mass functions (p.m.f.): $p(x)\triangleq \mathbf{Pr(}X=x\mathbf{)}$, $p(x,y)\triangleq \mathbf{Pr(}X=x,Y=y\mathbf{)}$, and $p(x|y)\triangleq \mathbf{Pr(}X=x|Y=y\mathbf{)}$. Finally, $[x]^+$ denotes $\max\{0,x\}$.

For simplicity in notion, we need the following definitions.
\begin{definition}\label{def:somefuncs}
\begin{align*}
\begin{split}
C_1(x) \buildrel \Delta \over = & \frac{1}{2}\log (1+x),\\
C_2(a,b,c,d,e,f,g,h) \buildrel \Delta \over = & \frac{1}{2}\log
\left({\frac{\bigl({|ac-b|^2df+ed+c^2fe}\bigr)\bigl({|ag+h|^2d+|g|^2e+|g+h|^2f+1}\bigr)}{|g(b-1)+h(c-1)|^2fde+|ac-b|^2df+ed+c^2fe}}\right).
\end{split}
\end{align*}
\end{definition}

\subsection{System Model}\label{subsec:sys-model}
To specify the discrete memoryless REC-SI depicted in Fig. \ref{fig:sys-model}, we define five sets $\left({\mathcal{X}_1,\mathcal{X}_2,\mathcal{Y}_2,\mathcal{Y},\mathcal{Z}}\right)$. We also define a transition probability distribution $p(y_2,y,z|x_1,x_2,s)$ for all\\ $\left({x_1,x_2,y_2,y,z,s\in \mathcal{X}_1,\mathcal{X}_2,\mathcal{Y}_2,\mathcal{Y},\mathcal{Z},\mathcal{S}}\right)$. $\mathcal{X}_1$ and $\mathcal{X}_2$ are the transmitter and relay inputs, respectively while $\mathcal{Y}_2$, $\mathcal{Y}$, and $\mathcal{Z}$ are the outputs of the relay, legitimate receiver, and eavesdropper, respectively. We suppose that the random state information $S$ is non-casually known to the transmitter and relay only, and taking values from a finite set $\mathcal{S}$. In this model, the transmitter wants to transmit a message $m \in \left\{{1,2,...,2^{nR} }\right\}$ to the receiver with the help of relay using a $\left({2^{nR},n,P_e^n}\right)$ code in $n$ channel uses. The relay helps the transmitter by relaying the message using DF scheme.

\begin{figure}[t!]
	\centering
	\includegraphics[width=0.75\linewidth]{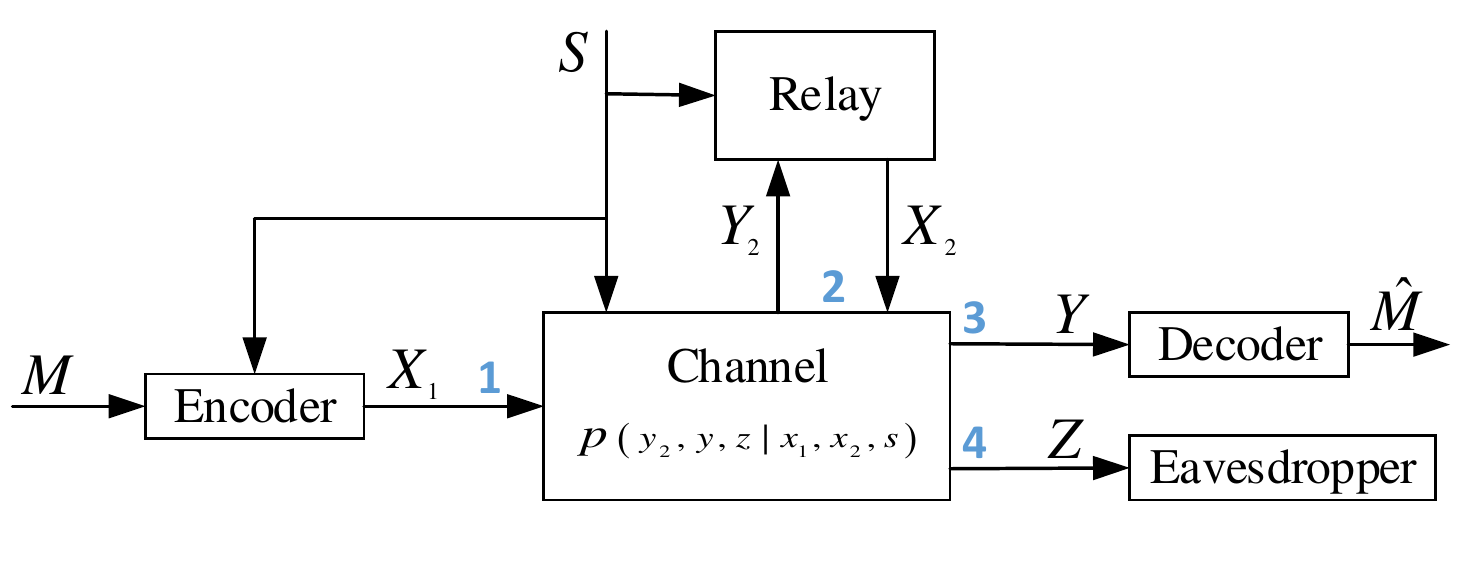}
	\caption{The four-terminal relay eavesdropper channel with state information.}
	\label{fig:sys-model}
\end{figure}

\begin{definition}
A $\left({2^{nR},n,P_e^n}\right)$ code for the discrete memoryless REC-SI consists of\\
\begin{enumerate}[I.]
\item A stochastic encoder function at the transmitter,
\begin{equation}
\Phi_1^n: \left\{{1,2,...,2^{nR}}\right\} \times \mathcal{S}^n \rightarrow \mathcal{X}_1^n.
\end{equation}
\item A sequence of stochastic encoder functions at the relay,
\begin{equation}
\Phi_2(i): \mathcal{Y}_2^{i-1}\times \mathcal{S}^n \rightarrow \mathcal{X}_2(i) \quad \text{for} \qquad i=1,2,...,n.
\end{equation}
\item A decoding function at the legitimate receiver,
\begin{equation}
\Psi^n: \mathcal{Y}^n \rightarrow \left\{{1,2,...,2^{nR}}\right\}.
\end{equation}
\item The average error probability of this code is 
\begin{equation}
P_e^n = \frac{1}{2^{nR}} \sum\limits_{m=1}^{2^{nR}}{\mathbf{Pr}\biggl\{{\Psi^n(Y^n) \neq m \left|{ \vphantom{\sum}}\right. m~\text{was sent}}\biggr\}}.
\end{equation}
\item The information leakage rate for the eavesdropper is measured by
\begin{equation}
R_L = \frac{1}{n}I(M;Z^n).
\end{equation}
\end{enumerate}
\end{definition}

\begin{definition}
The rate-leakage pair $(R,R_L^*)$ is achievable for the discrete memoryless REC-SI if for any $\epsilon >0$, there exist a sequence of coeds $\left({2^{nR},n,P_e^n}\right)$ such that for sufficiently large $n$, we have
\begin{equation}
P_e^n \leq \epsilon \quad \text{and} \quad R_L \leq R_L^* + \epsilon.
\end{equation}
The prefect secrecy rate $R_s$ is achievable if the rate-leakage pair $(R_s,R_s)$ is achievable.
\end{definition}

\section{Achievable Secrecy Rate for REC-SI}\label{sec:achrate}
In this section, we present an achievable perfect secrecy rate for the discrete memoryless REC-SI with DF strategy after the following definitions.
\begin{definition}\label{def1}
	Define $\mathcal{P}$ as the set of all joint distributions of the random variables $S$, $U_1$, $U_2$, $X_1$, $X_2$, $Y_2$, $Y$, and $Z$ which factor as
	\begin{equation}\label{eq:pmf}
		p(s,u_1,u_2,x_1,x_2,y_2,y,z)=p(s)p(u_1,u_2,x_1,x_2|s)p(y_2,y,z|x_1,x_2,s).
	\end{equation}
\end{definition}
\begin{definition}
	Given $p\in \mathcal{P}$, define the following rates:
	\begin{equation}
		\tilde{R} \buildrel \Delta \over =  \biggl[{I(U_1;Y_2|S,U_2)-I(U_1,U_2;Z)}\biggr]^+
	\end{equation}
	and
	\begin{equation}
		\hat{R} \buildrel \Delta \over = \biggl[{I(U_1,U_2;Y)-I(U_1,U_2;S)-I(U_1,U_2;Z)}\biggr]^+.
	\end{equation}
\end{definition}

An achievable perfect secrecy rate for discrete memoryless REC-SI is given by the following theorem.
\begin{theorem}\label{th1}
	For a discrete memoryless REC-SI with DF strategy and sate information $S^n$ non-causally available at the transmitter and relay, the following perfect secrecy rate is achievable:
	\begin{equation}\label{eq:achrate}
		R_s = \mathop {\sup }\limits_{p \in \mathcal{P}} 
		{\min\left\{{\tilde{R},\hat{R}}\right\} }.
	\end{equation}
\end{theorem}

\begin{proof}
	The proof of this theorem is based on random coding scheme which combines Csiszar {\it{et al.}} \cite{ref:csiszarKorner} and sliding-window decoding strategy \cite{ref:kingdiissertation,ref:carleial,ref:bookElgamal}. The proof is provided in Appendix \ref{app:proof1}.
\end{proof}

The presented achievable perfect secrecy rate for the discrete memoryless REC-SI is generalization of some previous works.

\begin{remark}
	If $Z=\varnothing$, our model reduces to the RC-SI and Theorem \ref{th1} reduces to the achievable rate presented in \cite[Theorem 2.1]{ref:zaidiRcSide}.
\end{remark}

\begin{remark}
	If $S=\varnothing$, our model reduces to the REC and Theorem \ref{th1} reduces to the perfect secrecy rate achieved in \cite[Theorem 2]{ref:elgamalREC}.
\end{remark}

\begin{remark}
	If we set $U_1=X_1$ and $U_2=X_2$ and disable the state information ($S=\varnothing$), our model is reduced to the classical RC and Theorem \ref{th1} is reduced to the rate of RC in \cite[Theorem 6]{ref:coverElgamalTh}.
\end{remark}

\section{The Gaussian REC-SI}\label{sec:gauss-ch}
In this section, we extend the results of Section \ref{sec:achrate} to the AWGN relay eavesdropper channel with additive state information. As previous section, we suppose that the state information is known to the transmitter and relay.

\subsection{Channel Model}\label{subsec:gaussmodel}
Assume that the received signals at each nodes of the model depicted in Fig. \ref{fig:sys-model} at times $i=1,2,\cdots,n$ are as follows:
\begin{align}
	\begin{split}
		Y_2(i) = & h_{sr}\bigl[{X_1(i) + S(i)}\bigr] + Z_1(i),\\
		Y(i) = & h_{sd}\bigl[{X_1(i)+S(i)}\bigr] + h_{rd}\bigl[{X_2(i)+S(i)}\bigr]+Z_2(i), \quad \mbox{and}\\  
		Z(i) = & h_{se}\bigl[{X_1(i)+S(i)}\bigr] + h_{re}\bigl[{X_2(i)+S(i)}\bigr]+Z_3(i)
	\end{split}
\end{align}
where $X_1(i)$ and $X_2(i)$ are the transmitted signals by the transmitter and relay, $S(i)$ is the channel state information and non-causally known to the transmitter and relay, $Y_2(i), Y(i)$ and $Z(i)$ are the outputs of the relay, legitimate receiver, and eavesdropper, respectively, $Z_1(i), Z_2(i)$ and $Z_3(i)$ are i.i.d. and independent white Gaussian noises with unit variance, and $h_{jk}$ is the channel coefficient between node $j$ and node $k$. We suppose that the channel state information at different times are i.i.d. and distributed as $\mathcal{N}(0,Q)$. We also let $X_2(i) \sim \mathcal{N}(0,P_2)$ and
\begin{equation}
	X_1(i) = cX_2(i)+\hat{X}(i)
\end{equation}
where $c$ is a constant and the novel information is modeled by $\hat{X}(i)$, i.i.d. $\mathcal{N}(0,P)$. The power constraint at the transmitter and relay is 
\begin{equation}
	|c|^2P_2+P\leq P_1.
\end{equation}

\subsection{Achievable Secrecy Rate for the Gaussian REC-SI}\label{subsec:gaussachrate}
The main result of Thoerem \ref{th1} for the Gaussian REC-SI is presented in the following theorem.
\begin{theorem}\label{th2}
	Using the auxiliary random variables $U_1=X_1+\alpha_1 S$ and $U_2=X_2+\alpha_2 S$, where $\alpha_1$ and $\alpha_2$ are real numbers and the state information $S$ is independent of $X_1$ and $X_2$, the following perfect secrecy rate is achievable for the Gaussian REC-SI:
	\begin{equation}\label{eq:gaussachrate}
		R_s^G = \mathop {\max }\limits_{c, \alpha_1, \alpha_2, P} \min\left\{{  \tilde{R}^G, \hat{R}^G }\right\},
	\end{equation}
	where
	\begin{equation}\label{eq:gaussachrate1}
		\tilde{R}^G = \biggl[{C_1\left({|h_{sr}|^2P}\right) - C_2 \left({c,\alpha_1,\alpha_2,P_2,P,Q,h_{se},h_{re}}\right)}\biggr]^+
	\end{equation}
	and
	\begin{equation}\label{eq:gaussachrate2}
			\hat{R}^G = \biggl[{C_2 \left({c,\alpha_1,\alpha_2,P_2,P,Q,h_{sd},h_{rd}}\right) - C_2 \left({c,\alpha_1,\alpha_2,P_2,P,Q,h_{se},h_{re}}\right) }\biggr. 
			 \biggl.{	- C_1 \left({\frac{|c\alpha_2-\alpha_1|^2P_2+\alpha_2^2P}{PP_2}Q}\right)}\biggr]^+.
	\end{equation}
\end{theorem}
\begin{proof}
	Please refer to Appendix \ref{app:proof2} where the proof is provided.
\end{proof}

\section{Numerical Examples}\label{sec:numex}
In this section, we discuss several numerical examples for the Gaussian REC-SI. The network topology is selected as suggested in \cite{ref:elgamalREC} and shown in Fig. \ref{fig:gaussmodel}. In this network, the transmitter, relay, receiver, and eavesdropper are placed at $(0,0)$, $(x,0)$, $(1,0)$, and $(1,0)$, respectively. In the following, we consider two modes for this channel, real and fading channel, as presented in \cite{ref:elgamalREC}. 

\begin{figure}[t!]
	\centering
	\includegraphics[width=0.6\linewidth]{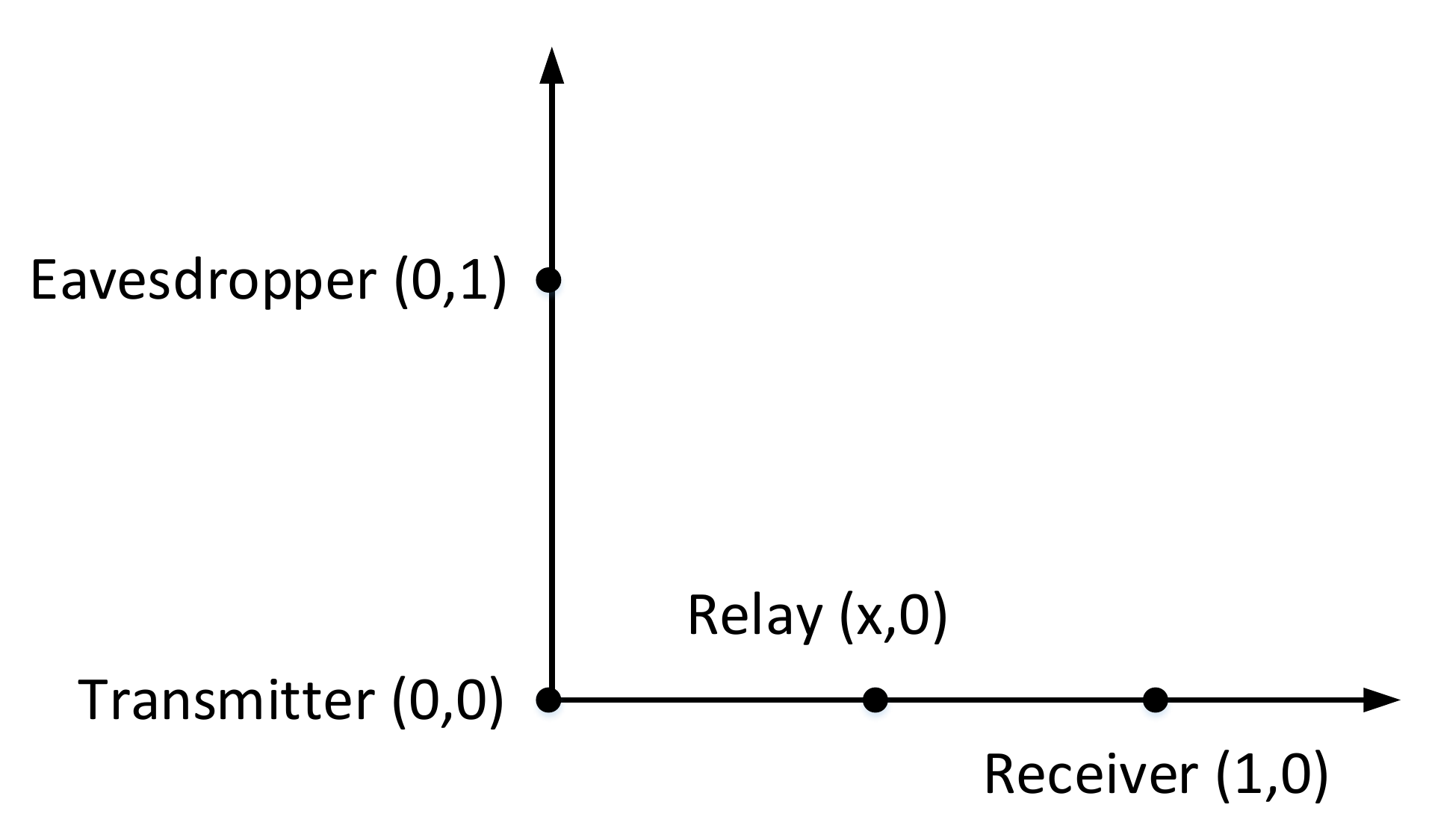}
	\caption{The network topology of Gaussian REC-SI.}
	\label{fig:gaussmodel}
\end{figure}

\subsection{Real Channel}\label{subsec:realch}
In this model, the channel is real; hence the channel coefficient is $h_{jk}=d_{jk}^{-\gamma}$ where $d_{jk}$ is the distance between node $j$ and $k$ and $\gamma>1$ is the channel attenuation coefficient. We give the achievable perfect secrecy rate of this channel under several scenarios. In all numerical examples of this subsection, we set $P_1=1$ and $P_2=8$, and $\gamma=3$.

In the first scenario, the power of state information, $Q$, is changed and the relay is moved along the $x$-axis.  Fig. \ref{fig:changeQ} illustrates the achievable perfect secrecy rate versus the position of the relay, $x$, with and without state information. The dashed curve belongs to the state independent REC. It can be seen that how the state information could help the secure communication between the transmitter and receiver. Moreover, it is not possible to get positive secrecy rate for $x>1$ in REC without state information while as seen in Fig. \ref{fig:changeQ}, we could get positive secrecy rate for this case. Finally, when $x=0$, we could not get positive secrecy, meaning using multi-antenna at the transmitter is ineffective.

\begin{figure}[t!]
	\centering
	\includegraphics[width=0.7\linewidth]{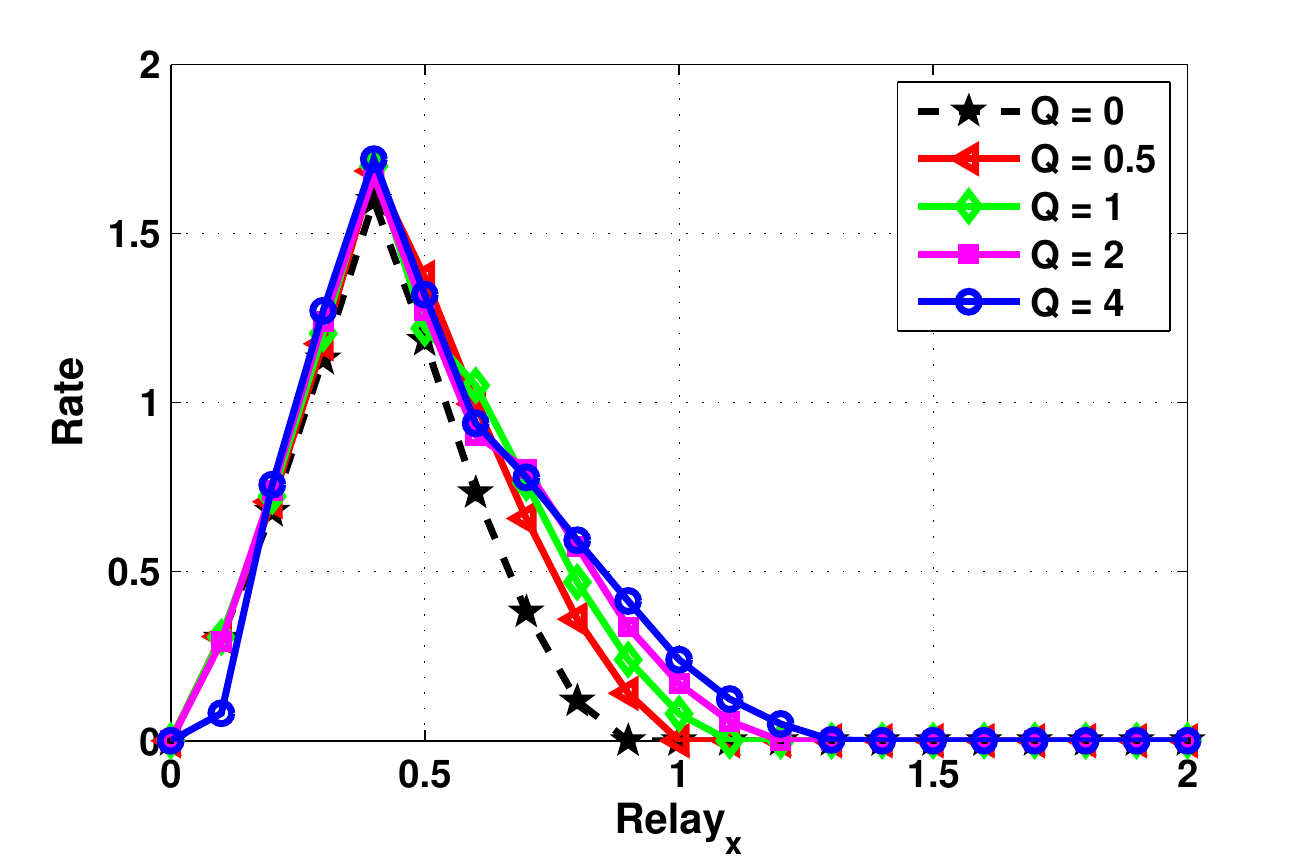}
	\caption{The achievable perfect secrecy rate for the Gaussian REC-SI with and without state information.}
	\label{fig:changeQ}
\end{figure}

In the second scenario, we set $Q=0.5$ and the eavesdropper is moved between the transmitter and point $(0,1)$ along the $y$-axis. The results of this scenario are depicted in Fig. \ref{fig:changeY}. In our network topology, not necessarily with the eavesdropper closer to the transmitter, the secrecy rate will not decrease. Even putting the eavesdropper nearby the transmitter, because it does not know about state information, we can get positive secrecy rate.

\begin{figure}[t!]
	\centering
	\includegraphics[width=0.7\linewidth]{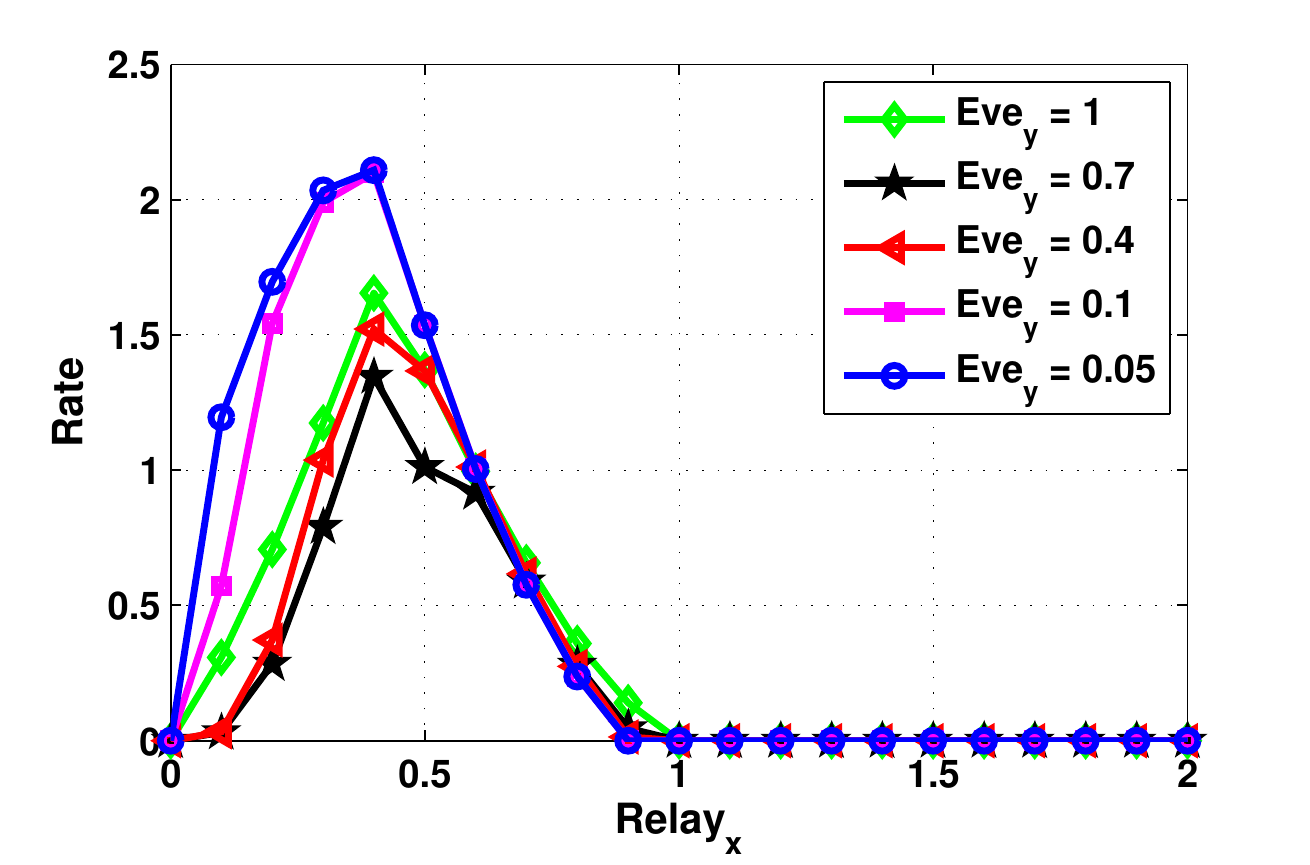}
	\caption{The achievable perfect secrecy rate for the Gaussian REC-SI when the eavesdropper is moved between the transmitter and $(0,1)$. $(Q=0.5)$}
	\label{fig:changeY}
\end{figure}

In the third scenario, the eavesdropper is nearer to the transmitter than the legitimate receiver. It is placed at the point $(0,0.4)$ and $Q$ is changed. Fig. \ref{fig:Y4} shows the achievable perfect secrecy of this scenario with and without state information. We can see how the state information could help these secure communications. Furthermore, we cannot get positive secrecy rate without state.

\begin{figure}[t!]
	\centering
	\includegraphics[width=0.7\linewidth]{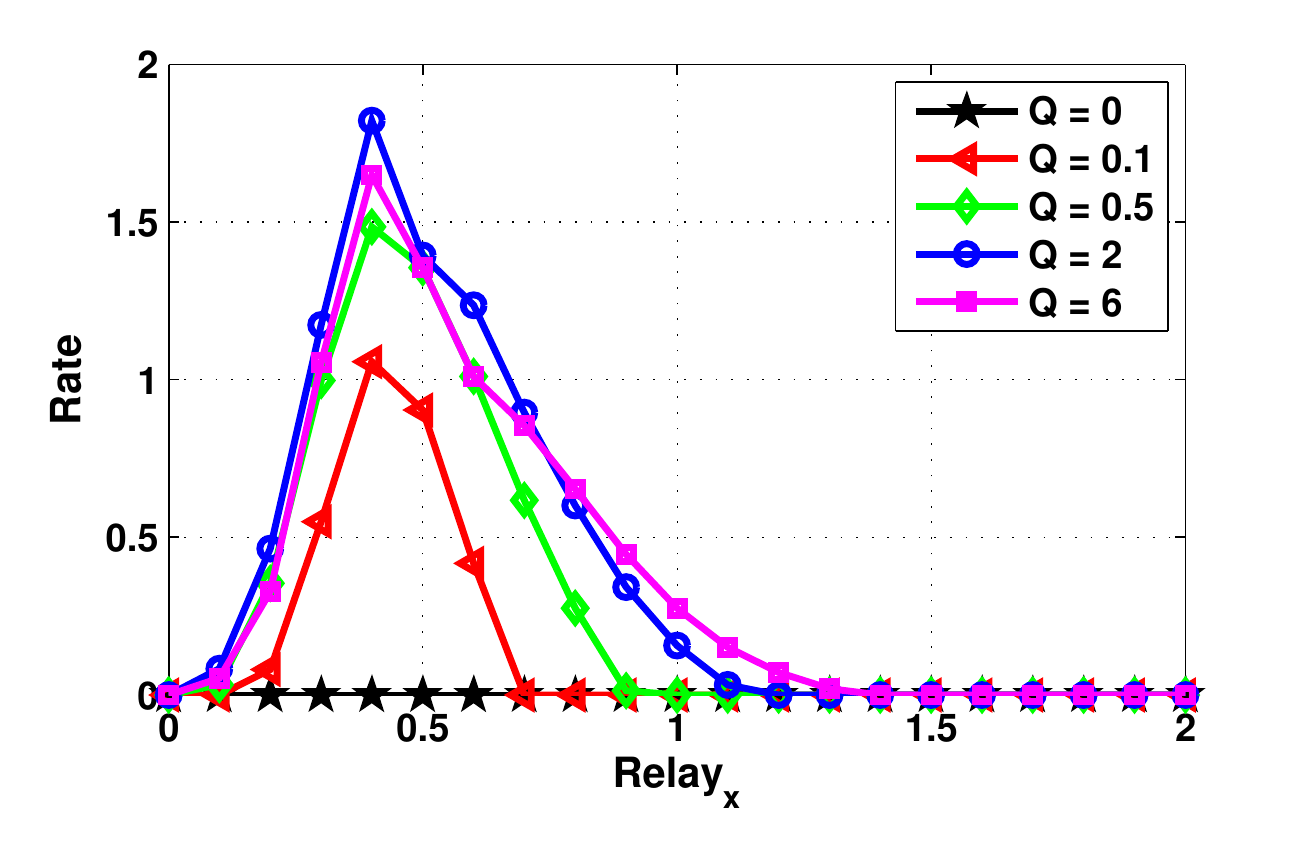}
	\caption{The achievable perfect secrecy rate for the Gaussian REC-SI with and without state information when the eavesdropper is placed at point $(0,0.4)$.}
	\label{fig:Y4}
\end{figure}

Finally, in the fourth scenario, the relay is moved on the straight line between the eavesdropper and receiver. The achievable secrecy rate is shown in Fig. \ref{fig:1-x}. Again, we cannot establish secure communication without the state information. But if we increase the power of the transmitter to $P_1=4$, we can also get positive secrecy rate without the state. This results of this scenario with $P_1=4$ and $P_2=8$ are shown in Fig. \ref{fig:p14p28}.

\begin{figure}[t!]
	\centering
	\includegraphics[width=0.7\linewidth]{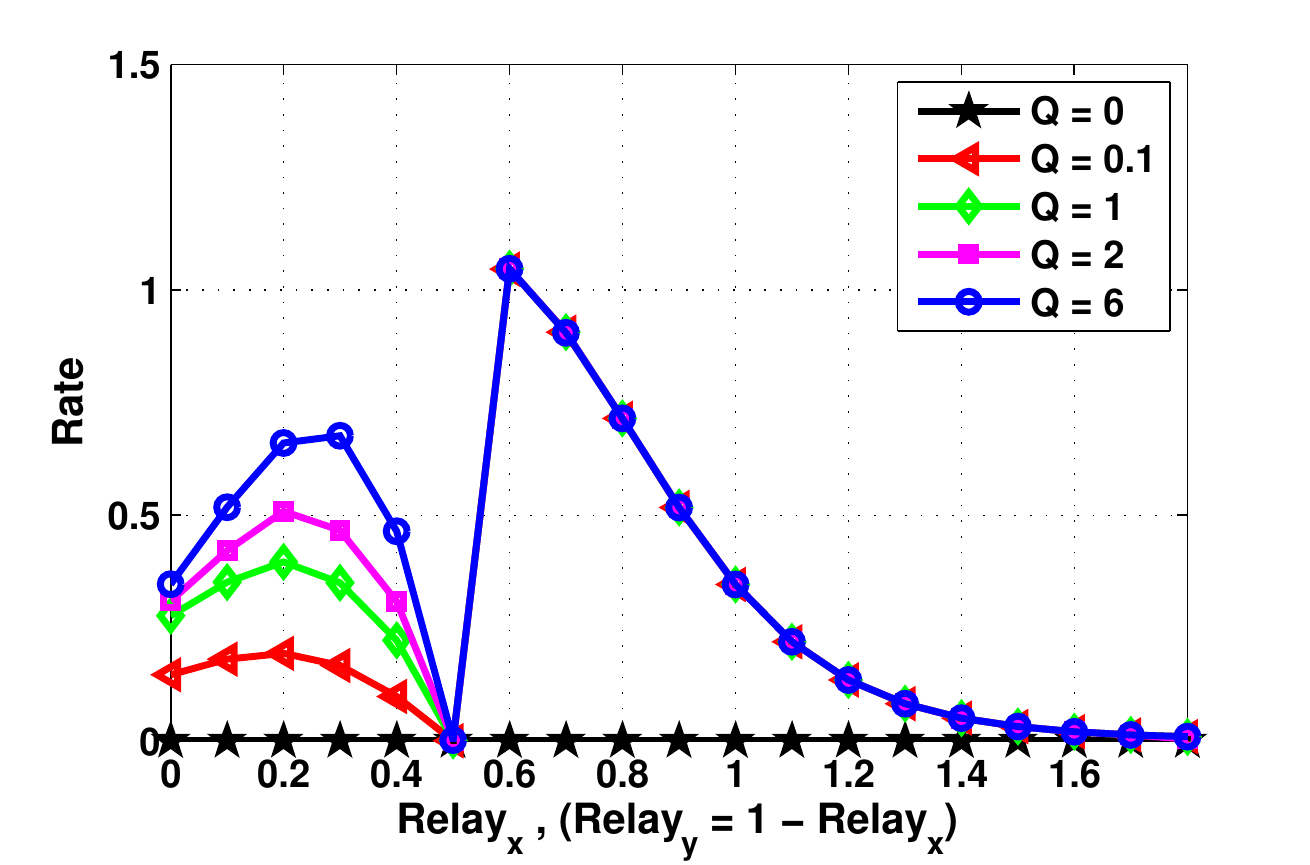}
	\caption{The achievable perfect secrecy rate for the Gaussian REC-SI when the relay is moved on the straight line between eavesdropper and receiver.}
	\label{fig:1-x}
\end{figure}

\begin{figure}[t!]
	\centering
	\includegraphics[width=0.7\linewidth]{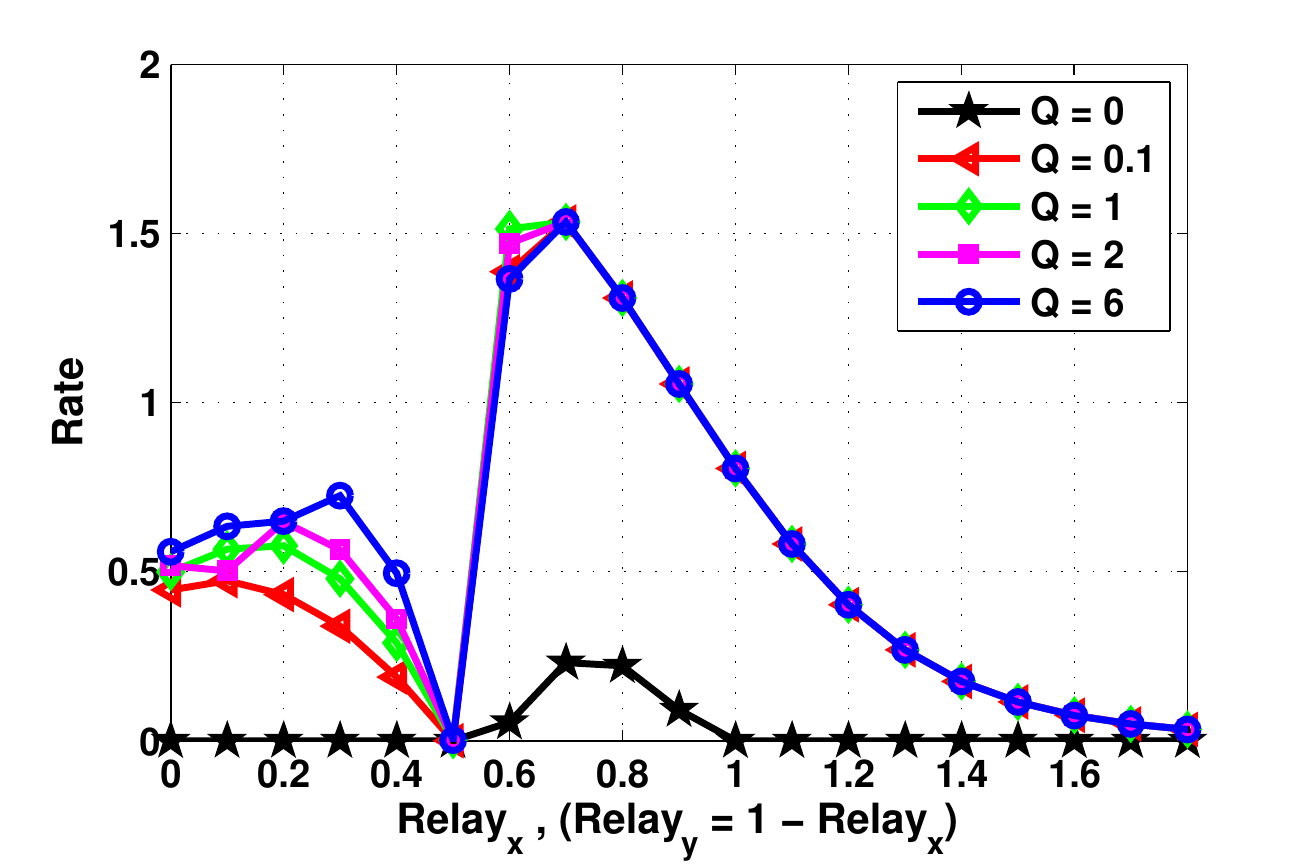}
	\caption{The achievable perfect secrecy rate for the Gaussian REC-SI when the relay is moved on the straight line between eavesdropper and receiver. ($P_1=4$ and $P_2=8$)}
	\label{fig:p14p28}
\end{figure}

\subsection{Fading Channel}\label{subsec:fadech}
In this model we consider a random phase for each channel coefficients, i.e., $h_{jk}=d_{jk}^{-\gamma}e^{j\theta_{j,k}}$. In this model, the transmitter knows the phases $\theta_{12}$, $\theta_{13}$, $\theta_{23}$ but does not know the $\theta_{14}$ and $\theta_{24}$. We choose the fading phases uniformly from $[0,2\pi)$. Moreover, $\theta_{13}$ and $\theta_{13}$ are independent of other fading phases. Here, we only consider the first scenario presented in Subsection \ref{subsec:realch} to see how the state information also helps secure communication in Gaussian fading channel. In this example, we also set $P1=1$, $P_2=8$, and $\gamma=3$. The results are shown in Fig. \ref{fig:fadech}.

\begin{figure}[t!]
	\centering
	\includegraphics[width=0.7\linewidth]{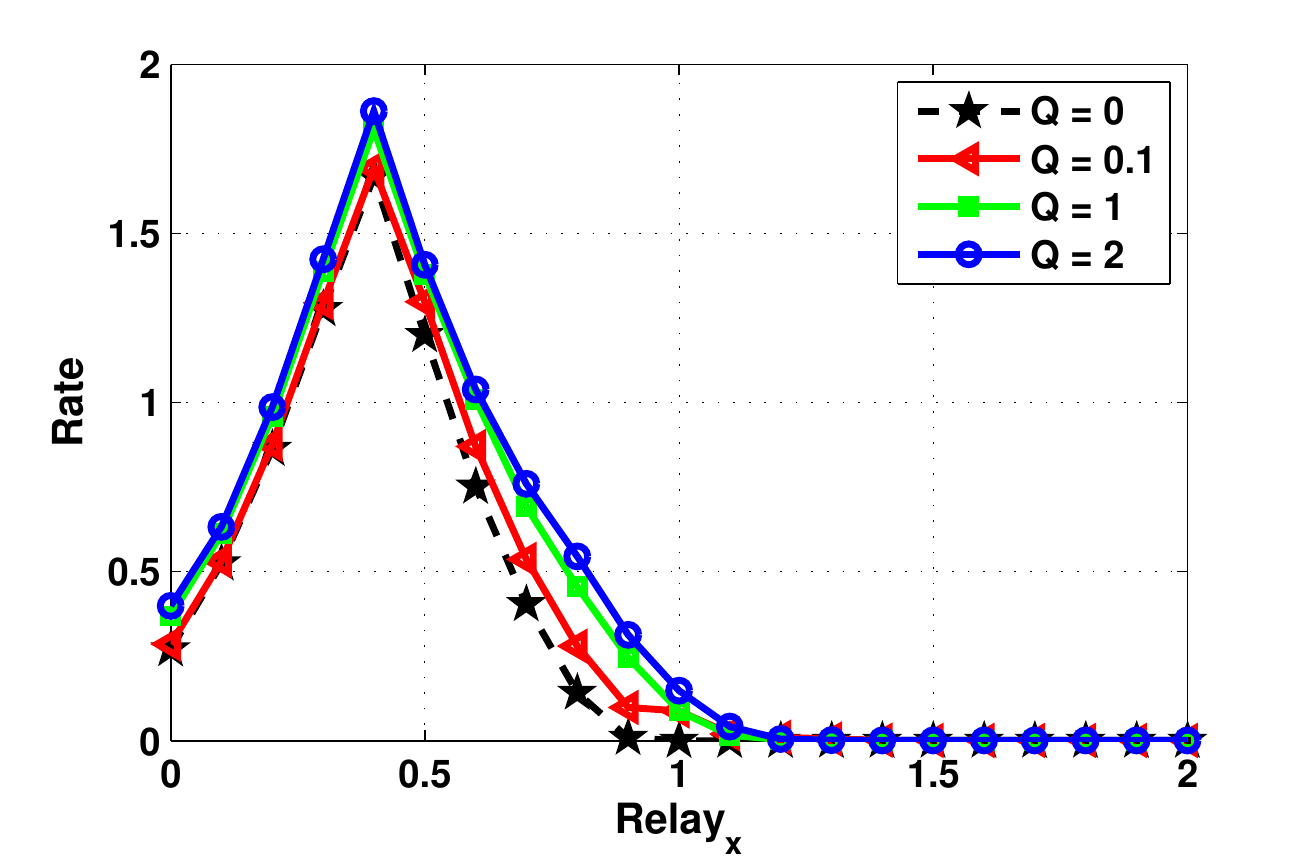}
	\caption{The achievable perfect secrecy rate for the fading Gaussian REC-SI.}
	\label{fig:fadech}
\end{figure}

We can see that state is also helpful in this example. Furthermore, we can see that using multi-antenna in the transmitter is effective in fading channel.

\section{Conclusion}\label{sec:con}
In this paper, we give an achievable secrecy rate for the memoryless REC-SI. The presented model was a generalization of some existing models, where the transmitter wishes to transmit a confidential message to the legitimate receiver keeping secret from an eavesdropper with the help of a relay as an active cooperating partner. Moreover, a random state is non-causally available only to the transmitter and relay. The relay helps this secret communication by relaying the message using DF scheme. We also extended our model to the AWGN channel with an additive Gaussian random state information. An achievable secrecy rate was also derived for the Gaussian REC-SI. We provided some numerical examples, where it could be seen how the state information helps to achieve secure communication.

\appendices
\section{Proof of Theorem \ref{th1}}\label{app:proof1}
\begin{proof}
As stated before, the proof is based on random coding scheme and combination of Csiszar {\it{et al.}} \cite{ref:csiszarKorner} and sliding-window decoding strategy \cite{ref:kingdiissertation,ref:carleial,ref:bookElgamal}. Let define the following parameters for simplicity:

\begin{definition}\label{defrates}
	\begin{align}
		\begin{split}\label{eq:r1}
			R_1 \buildrel \Delta \over = I(U_1,U_2;Z)+\epsilon, \\
		\end{split}\\
		\begin{split}\label{eq:r2}
			R_2 \buildrel \Delta \over = I(U_2;S)+\epsilon, \\
		\end{split}\\
		\begin{split}\label{eq:r3}
			R_3 \buildrel \Delta \over = I(U_1;S|U_2)+\epsilon,\\
		\end{split}\\\
		\begin{split}\label{eq:rate}
			R \buildrel \Delta \over = \min\biggl\{{   I(U_1;Y_2|S,U_2),  I(U_1,U_2;Y)-I(U_1,U_2;S)     }\biggr\} -  I(U_1,U_2;Z)  - 6\epsilon.\\
		\end{split}
	\end{align}
\end{definition}

Consider $B$ transmission blocks in which $n$ transmissions occur. A sequence of $(B-1)$ messages $M_b$, $b\in[1:B-1]$, each selected independently and uniformly over $[1:2^{nR}]$ is sent over these $B$ blocks. Hence the average rate over these blocks is $R(B-1)/B$ which can be made arbitrary close to $R$ as $n \rightarrow \infty$.

{\it{Codebook generation:}}  Fix a p.m.f. as (\ref{eq:pmf}) that attains the achievable perfect secrecy rate in (\ref{eq:achrate}). For each block $b\in [1:B]$, generate $2^{n(R+R_1+R_2)}$ i.i.d. $u^n_2$ sequences, each of length $n$ according to probability $\prod\limits_{i = 1}^n {p({u_{2i}})}$. Index them as $u^n_2(m_{b-1},l_{b-1},w'_{b-1})$ where $m_{b-1}\in[1:2^{nR}]$, $l_{b-1}\in[1:2^{nR_1}]$, $w'_{b-1}\in[1:2^{nR_2}]$. Moreover, for each $u^n_2$, generate $2^{n(R+R_1+R_3)}$ conditionally i.i.d. $u^n_1$ sequences with $\prod\limits_{i = 1}^n {p({u_{1i}}|{u_{2i}})}$. Also index them as $u^n_1(m_{b-1},l_{b-1},w'_{b-1},m_b,l_{b},w_{b})$ where $m_{b}\in[1:2^{nR}]$, $l_{b}\in[1:2^{nR_1}]$, $w_{b}\in[1:2^{nR_3}]$.

{\it{Encoding:}} Let $m_b\in[1:2^{nR}]$ be the message to be sent in block $b$ with state information $s^n$ observed. At first, the stochastic encoder at the transmitter looks in bin $m_{b-1}$ and subbin $l_{b-1}$ for the smallest index $w'_{b-1}$ such that $\left({u^n_2(m_{b-1},l_{b-1},w'_{b-1}),s^n}\right)\in T^{(n)}_{\epsilon}$. Then, it uniformly randomly chooses $l_b$ from bin $m_b$ and looks in subbin $l_{b}$ for a sequence $u^n_1$ such that\\
$\left({u^n_1(m_{b-1},l_{b-1},w'_{b-1},m_{b},l_{b},w_b),u^n_2(m_{b-1},l_{b-1},w'_{b-1}),s^n}\right)\in T^{(n)}_{\epsilon}$ and transmits $\prod\limits_{i = 1}^n {p({x_{1i}}|{u_{2i}},{s_i})}$. If there is no such sequences, then the sender randomly chooses one of them. Assume $m_0=m_B=1$ by convention. In block ($b+1$), the relay knows the estimates $(\tilde m_b,\tilde l_b)$ of the message sent by the transmitter in the previous block; hence, it looks in bin $\tilde m_b$ and subbin $\tilde l_b$ for the smallest index $\tilde w'_b$ such that $\left({u^n_2(\tilde m_b,\tilde l_b,\tilde w'_b),s^n}\right)\in T^{(n)}_{\epsilon}$ and sends $\prod\limits_{i = 1}^n {p({x_{2i}}|{u_{2i}},{s_i})}$.

{\it{Decoding:}} Let $\tilde m_0=1$ by convention. At the end of block $b$, the relay finds the unique $(\tilde m_b,\tilde l_b)$ such that $$\left({u^n_1(m_{b-1},l_{b-1},w'_{b-1},\tilde m_b,\tilde l_b,\tilde w_b),u^n_2(m_{b-1},l_{b-1},w'_{b-1}),s^n,y^n_2(b)}\right)\in T^{(n)}_{\epsilon}$$ for some $\tilde w_b$.
At the end of block ($b+1$), the receiver finds the unique $(\hat m_b,\hat l_b)$ such that $$\left({u^n_1(m_{b-1},l_{b-1},w'_{b-1},\hat m_b,\hat l_b,\hat w_b),u^n_2(m_{b-1},l_{b-1},w'_{b-1}),y^n(b)}\right)\in T^{(n)}_{\epsilon}$$ for some $\hat w_b$ and $\left({u^n_2(\hat m_b,\hat l_b,\hat w'_b),y^n(b+1)}\right)\in T^{(n)}_{\epsilon}$ for some $\hat w'_b$.

{\it{Analysis of the error probability:}} Assume without loss of generality that the transmitted messages are $M_{b-1}=M_{b}=1$ with the corresponding indices $L_{b-1}$ and $L_{b}$ and $\tilde M_b$ and $\hat M_b$ are the relay and receiver message estimates, respectively. If one or more of the following events occur the decoder makes an error.
\begin{align}
\begin{split}\label{eq:eq113}
E^1_1(b)=\left\{{         
\left({ u^n_1(M_{b-1},L_{b-1},w'_{b-1},M_b,L_b,w_b) , u^n_2(M_{b-1},L_{b-1},w'_{b-1}),s^n}\right) \notin T^{(n)}_{\epsilon}{\rm~for~all~}w_b}\right\}, \\
\end{split}\\
\begin{split}\label{eq:e12}
E^1_2(b)=\left\{{         
\left({u^n_2(M_{b-1},L_{b-1},w'_{b-1}),s^n}\right) \notin T^{(n)}_{\epsilon}{\rm~for~all~}w'_{b-1}}\right\}, \\
\end{split}\\
\begin{split}\label{eq:e21}
E^2_1(b+1)=\left\{{         
\left({u^n_2(\tilde M_b,\tilde L_b,\tilde w'_b),s^n}\right) \notin T^{(n)}_{\epsilon}{\rm~for~all~}\tilde w'_b}\right\},\\
\end{split}\\
\begin{split}\label{eq:e221}
E^2_2(b-1)=\left\{{ \tilde M_{b-1} \neq 1}\right\},\\
\end{split}\\
\begin{split}\label{eq:e222}
E^2_2(b)=\left\{{ \tilde M_b \neq 1}\right\},\\
\end{split}\\
\begin{split}\label{eq:e}
E(b-1)=\left\{{ \hat M_{b-1} \neq 1}\right\},\\
\end{split}\\
\begin{split}\label{eq:e31}
E^3_1(b)=\left\{{         
\left({ u^n_1(\hat M_{b-1},\hat L_{b-1},\hat W'_{b-1},\tilde M_b,\tilde L_b,W_b) , u^n_2(\tilde M_{b-1},\tilde L_{b-1},\tilde W'_{b-1}),y^n(b)}\right) \notin T^{(n)}_{\epsilon} }\right.\\ 
\left.{~{\rm or}~\left({u^n_2(\tilde M_b,\tilde L_b,\tilde W'_b),y^n(b+1)}\right) \notin T^{(n)}_{\epsilon}}\right\}, \\
\end{split}\\
\begin{split}\label{eq:e32}
E^3_2(b)=\left\{{
\left({ u^n_1(\hat M_{b-1},\hat L_{b-1},\hat W'_{b-1},m_b,l_b,w_b) , u^n_2(\hat M_{b-1},\hat L_{b-1},\hat W'_{b-1}),y^n(b)}\right) \in T^{(n)}_{\epsilon} }\right.\\ 
~{\rm and}~\left({u^n_2(m_b,l_b,w'_b),y^n(b+1)}\right) \in 
\left.{ T^{(n)}_{\epsilon} {\rm~for~some~} (m_b,l_b)\neq (\tilde M_b,\tilde L_b)}\right\}. \\
\end{split}
\end{align}

The probability of error can be upper bounded as 
\begin{align}
\begin{split}\label{eq:eq114}
P\left({E(b)}\right) = P\left\{{   \hat M_b \neq 1 }\right\} \leq {}& P\left\{{  E^1_1(b) }\right\} + P\left\{{  E^1_2(b) }\right\} + P\left\{{  E^2_1(b+1) }\right\}  + P\left\{{  E^2_2(b-1) }\right\} \\ 
{}&  + P\left\{{  E^2_2(b) }\right\}+ P\left\{{  E(b-1) }\right\}+ P\left\{{  E^3_1(b) \cap  E^{2c}_2(b-1) \cap  E^{2c}_2(b) \cap  E^c(b-1)}\right\} \\
{}&  + P\left\{{  E^3_2(b)\cap  E^{2c}_2(b) \cap  E^c(b-1)}\right\}.
\end{split}
\end{align}
The first, second, and third terms tends to zero as $n\rightarrow \infty$ since (\ref{eq:r2}) and (\ref{eq:r3}), respectively, imply that $R_2 > I(U_2;S)$ and $R_3 > I(U_1;S|U_2)$. Since the definitions of $R$, $R_1$, and $R_3$ in Definition \ref{defrates} imply that $R+R_1+R_3 < I(U_1;S,Y_2|U_2)$, it is not difficult to prove that the forth, fifth, and seventh terms also tends to zero as $n\rightarrow \infty$.
For the eighth term, we have
{\small{
\begin{align}
\begin{split}
P&\left\{{\left({ u^n_1(\hat M_{b-1},\hat L_{b-1},\hat W'_{b-1},m_b,l_b,w_b) , u^n_2(\hat M_{b-1},\hat L_{b-1},\hat W'_{b-1}),y^n(b)}\right) \in T^{(n)}_{\epsilon}~{\rm and}~\left({u^n_2(m_b,l_b,w'_b),y^n(b+1)}\right) \in T^{(n)}_{\epsilon} }\right.\\
&\left.{{\rm~for~some~} (m_b,l_b)\neq (M_b, L_b) {\rm{~and~}}(\tilde M_b,\tilde L_b)= (M_b, L_b)  }\right\}  \\
{}&=\sum\limits_{{m_b} \ne {M_b}} { \sum\limits_{{l_b} \ne {L_b}}{     \left[{     P\left\{      \left({ u^n_1(\hat M_{b-1},\hat L_{b-1},\hat W'_{b-1},m_b,l_b,w_b) , u^n_2(\hat M_{b-1},\hat L_{b-1},\hat W'_{b-1}),y^n(b)}\right) \in T^{(n)}_{\epsilon}~{\rm and}~ (\tilde M_b,\tilde L_b)= (M_b, L_b) \right \}\times  }\right.  }} \\
&\phantom{=\sum\limits_{{m_b} \ne {M_b}}{\sum\limits_{{m_b} \ne {M_b}}{}}} {{  \left.{
P\left\{{\left({u^n_2(m_b,l'_b),y^n(b+1)}\right) \in T^{(n)}_{\epsilon}~| ~(\tilde M_b,\tilde L_b)= (M_b, L_b)  }\right\}   }\right] }}\\
{}&=\sum\limits_{{m_b} \ne {M_b}} { \sum\limits_{{l_b} \ne {L_b}}{     \left[{     P\left\{      \left({ u^n_1(\hat M_{b-1},\hat L_{b-1},\hat W'_{b-1},m_b,l_b,w_b) , u^n_2(\hat M_{b-1},\hat L_{b-1},\hat W'_{b-1}),y^n(b)}\right) \in T^{(n)}_{\epsilon} \right \}\times  }\right.  }} \\
&\phantom{=\sum\limits_{{m_b} \ne {M_b}}{\sum\limits_{{m_b} \ne {M_b}}{}}}  {{  \left.{
P\left\{{\left({u^n_2(m_b,l'_b),y^n(b+1)}\right) \in T^{(n)}_{\epsilon}~| ~(\tilde M_b,\tilde L_b)= (M_b, L_b)  }\right\}   }\right] }}\\
{}&\leq \sum\limits_{{m_b} \ne {M_b}} { \sum\limits_{{l_b} \ne {L_b}}{{2^{ - n\left[ {I\left( {{U_1};{Y}|{U_2}} \right) - 3\epsilon } \right]}}{2^{ - n\left[ {I\left( {{U_2};{Y}} \right) - 3\epsilon } \right]}}} }\leq {2^{n\left[ {R + R_1 + R_2 + R_3} \right]}}{2^{ - n\left[ {I\left( {{U_1},{U_2};{Y}} \right) - 6\epsilon } \right]}};
\end{split}
\end{align}
}}	

hence it is necessary to have 
\begin{equation}\label{eq:condall}
R+R_1+R_2+R_3 \leq I( U_1,U_2;Y) - 6\epsilon.
\end{equation}
According to Definition \ref{defrates}, we can see that (\ref{eq:condall}) holds and therefore, the eighth term also tends to zero as $n\rightarrow \infty$. Finally using induction the sixth term also tends to zero as $n\rightarrow \infty$.

{\it{Analysis of information leakage rate:}} 
We calculate the mutual information between message $M$ and $Z^n$, averaged over the random codebook $\mathcal{C}$.
\begin{align}\label{eq:leakage}
\begin{split}
I(M;Z^n|\mathcal{C})= {}& H(M|\mathcal{C})-H(M|Z^n,\mathcal{C}) = nR - H(M,l|Z^n,\mathcal{C})+H(l|M,Z^n,\mathcal{C}) \\ 
={}& nR - H(M,l|\mathcal{C}) + I(M,l;Z^n|\mathcal{C}) + H(l|M,Z^n,\mathcal{C}) \\
={}& nR - n\bigl({R+R_1}\bigr) + I(M,l,U^n_1,U^n_2;Z^n|\mathcal{C}) + H(l|M,Z^n,\mathcal{C}) \\
\le{}& -nR_1 + I(M,l,U^n_1,U^n_2,\mathcal{C};Z^n) + H(l|M,Z^n,\mathcal{C}) \\ 
\buildrel (a) \over \le{}& -nR_1 + I(U^n_1,U^n_2;Z^n)+ H(l|M,Z^n,\mathcal{C}) \\ 
\buildrel (b) \over \le {}& -nR_1 + nI(U_1,U_2;Z)+ H(l|M,Z^n,\mathcal{C}) + n\epsilon'\\ 
\buildrel (c) \over \le {}& -nR_1 + nI(U_1,U_2;Z)+ n\bigl({R_1-I(U_1,U_2;Z) + \epsilon}\bigr) \\
= {}& n\epsilon
\end{split}
\end{align}
(a) follows from the fact that $\mathcal{C},M,l  \rightarrow U^n_1,U^n_2  \rightarrow Z^n$ forms a Markov chain and (b) is due to the following fact that 
\begin{equation}\label{eq:unleqnu}
I(U^n_1,U^n_2;Z^n) \leq nI(U_1,U_2;Z)+n\epsilon'.
\end{equation}
The proof of (\ref{eq:unleqnu}) is similar to the proof provided in \cite{ref:Liu} (Lemma 3) and in \cite{ref:Xu} (Appendix A). (c) holds since by using Lemma 22.1 in \cite{ref:bookElgamal}, if we have $R_1 > I(U_1,U_2;Z)$, then $H(l|M,Z^n,\mathcal{C}) \le n\bigl({R_1 - I(U_1,U_2;Z) + \epsilon}\bigr)$. From (\ref{eq:r1}), we can see that this condition is satisfied and hence the proof is completed.
\end{proof}

\section{Proof of Theorem \ref{th2}}\label{app:proof2}
\begin{proof}
We extend the achievable perfect secrecy rate derived in Theorem \ref{th1} to the Gaussian case with continuous alphabets \cite{ref:bookCover} for the channel model presented in Theorem \ref{subsec:gaussmodel}. Using an appropriate choice of input distributions, it is sufficient to calculate the mutual information terms in (\ref{eq:achrate}). Straightforward calculations result in the following expressions:
\begin{equation}\label{eq:isr}
I(U_1;Y_2|S,U_2) = \frac{1}{2}\log
\left({1+|h_{sr}|^2P}\right),
\end{equation}
\begin{equation}\label{eq:ius}
I(U_1,U_2;S) = \frac{1}{2}\log
\left({1+\frac{|c\alpha_2-\alpha_1|^2P_2+\alpha_2^2P}{PP_2}Q}\right),
\end{equation}
\begin{equation}\label{eq:iuy}
I(U_1,U_2;Y) = \frac{1}{2}\log
\left({\frac{\left({|c\alpha_2-\alpha_1|^2P_2Q+PP_2+\alpha_2^2PQ}\right)\left({|ch_{sd}+h_{rd}|^2P_2+|h_{sd}|^2P+|h_{sd}+h_{rd}|^2Q+1}\right)}{|h_{sd}(\alpha_1-1)+h_{rd}(\alpha_2-1)|^2PP_2Q+|c\alpha_2-\alpha_1|^2P_2Q+PP_2+\alpha_2^2PQ}}\right),
\end{equation}
and
\begin{equation}\label{eq:iuz}
I(U_1,U_2;Z) = \frac{1}{2}\log
\left({\frac{\left({|c\alpha_2-\alpha_1|^2P_2Q+PP_2+\alpha_2^2PQ}\right)\left({|ch_{se}+h_{re}|^2P_2+|h_{se}|^2P+|h_{se}+h_{re}|^2Q+1}\right)}{|h_{se}(\alpha_1-1)+h_{re}(\alpha_2-1)|^2PP_2Q+|c\alpha_2-\alpha_1|^2P_2Q+PP_2+\alpha_2^2PQ}}\right).
\end{equation}
The proof is completed by substituting (\ref{eq:isr})-(\ref{eq:iuz}) in (\ref{eq:achrate}).
\end{proof}


\begin{thebibliography}{10}
	\providecommand{\url}[1]{#1}
	\csname url@samestyle\endcsname
	\providecommand{\newblock}{\relax}
	\providecommand{\bibinfo}[2]{#2}
	\providecommand{\BIBentrySTDinterwordspacing}{\spaceskip=0pt\relax}
	\providecommand{\BIBentryALTinterwordstretchfactor}{4}
	\providecommand{\BIBentryALTinterwordspacing}{\spaceskip=\fontdimen2\font plus
		\BIBentryALTinterwordstretchfactor\fontdimen3\font minus
		\fontdimen4\font\relax}
	\providecommand{\BIBforeignlanguage}[2]{{%
			\expandafter\ifx\csname l@#1\endcsname\relax
			\typeout{** WARNING: IEEEtran.bst: No hyphenation pattern has been}%
			\typeout{** loaded for the language `#1'. Using the pattern for}%
			\typeout{** the default language instead.}%
			\else
			\language=\csname l@#1\endcsname
			\fi
			#2}}
	\providecommand{\BIBdecl}{\relax}
	\BIBdecl
	
	\bibitem{ref:shannonSecSys}
	C.~E. Shannon, ``Communication theory of secrecy systems,'' \emph{The Bell
		System Technical Journal}, vol.~28, pp. 656--715, October 1949.
	
	\bibitem{ref:wtWyner}
	A.~D. Wyner, ``The wire-tap channel,'' \emph{Bell Syst. Tech. J.}, vol.~54,
	no.~8, pp. 1355--1387, 1975.
	
	\bibitem{ref:csiszarKorner}
	I.~Csiszar and J.~Korner, ``Broadcast channels with confidential messages,''
	\emph{IEEE Transaction on Information Theory}, vol.~24, no.~3, pp. 339--348,
	May 1978.
	
	\bibitem{ref:shannonSideInf}
	C.~E. Shannon, ``Channels with side information at the transmitter,'' \emph{IBM
		Journal of Research and Development}, vol.~2, no.~4, pp. 289--293, October
	1958.
	
	\bibitem{ref:gelfandPinsker}
	S.~I. Gel'fand and M.~S. Pinsker, ``Coding for channel with random
	parameters,'' \emph{Problems of Control and Information Theory}, vol.~9,
	no.~1, pp. 19--31, November 1980.
	
	\bibitem{ref:mitrpantGaussWtSide}
	C.~Mitrpant, A.~J.~H. Vinck, and L.~Yuan, ``An achievable region for the
	{Gaussian} wiretap channel with side information,'' \emph{IEEE Transactions
		on Information Theory}, vol.~52, no.~5, pp. 2181--2190, May 2006.
	
	\bibitem{ref:chen2008}
	Y.~Chen and A.~J.~H. Vinck, ``Wiretap channel with side information,''
	\emph{IEEE Transactions on Information Theory}, vol.~54, no.~1, pp. 395--402,
	January 2008.
	
	\bibitem{ref:elgamalCausalWt}
	C.~Yeow-Khiang and A.~E. Gamal, ``Wiretap channel with causal state
	information,'' in \emph{2010 IEEE International Symposium on Information
		Theory Proceedings (ISIT)}, June 2010, pp. 2548--2552.
	
	\bibitem{ref:khistiKeyAgr}
	A.~Khisti, S.~N. Diggavi, and G.~W. Wornell, ``Secret-key agreement with
	channel state information at the transmitter,'' \emph{IEEE Transactions on
		Information Forensics and Security}, vol.~6, no.~3, pp. 672--681, September
	2011.
	
	\bibitem{ref:BocheWtStrngSecrecy}
	H.~Boche and R.~F. Schaefer, ``Wiretap channels with side information--strong
	secrecy capacity and optimal transceiver design,'' \emph{IEEE Transactions on
		Information Forensics and Security}, vol.~8, no.~8, pp. 1397--1408, August
	2013.
	
	\bibitem{ref:villardSecTrns}
	J.~Villard, P.~Piantanida, and S.~Shamai, ``Secure transmission of sources over
	noisy channels with side information at the receivers,'' \emph{IEEE
		Transactions on Information Theory}, vol.~60, no.~1, pp. 713--739, January
	2014.
	
	\bibitem{ref:zibaeenejad}
	A.~Zibaeenejad, ``Key generation over wiretap models with non-causal side
	information,'' \emph{IEEE Transactions on Information Forensics and
		Security}, vol.~10, no.~7, pp. 1456--1471, July 2015.
	
	\bibitem{ref:ekremCoopSec}
	S.~Ulukus and E.~Ekrem, ``Cooperative secrecy in wireless communications,'' in
	\emph{Securing Wireless Communications at the Physical Layer}, R.~Liu and
	W.~Trappe, Eds.\hskip 1em plus 0.5em minus 0.4em\relax Springer US, 2010, pp.
	143--172.
	
	\bibitem{ref:coopJamming}
	E.~Tekin and A.~Yener, ``The general gaussian multiple-access and two-way
	wiretap channels: Achievable rates and cooperative jamming,'' \emph{IEEE
		Transactions on Information Theory}, vol.~54, no.~6, pp. 2735--2751, June
	2008.
	
	\bibitem{ref:artfNoise}
	L.~Ruoheng, I.~Maric, P.~Spasojevic, and R.~D. Yates, ``Discrete memoryless
	interference and broadcast channels with confidential messages: Secrecy rate
	regions,'' \emph{IEEE Transactions on Information Theory}, vol.~54, no.~6,
	pp. 2493--2507, June 2008.
	
	\bibitem{ref:elgamalREC}
	L.~Lifeng and H.~E. Gamal, ``The relay--eavesdropper channel: Cooperation for
	secrecy,'' \emph{IEEE Transactions on Information Theory}, vol.~54, no.~9,
	pp. 4005--4019, September 2008.
	
	\bibitem{ref:relayWt}
	M.~Yuksel and E.~Erkip, ``The relay channel with a wire-tapper,'' in \emph{41st
		Annual Conference on Information Sciences and Systems, 2007}, March 2007, pp.
	13--18.
	
	\bibitem{ref:xiangUnHelper}
	X.~He and A.~Yener, ``On the equivocation region of relay channels with
	orthogonal components,'' in \emph{41st Asilomar Conference on Signals,
		Systems and Computers, 2007}, November 2007, pp. 883--887.
	
	\bibitem{ref:elgamalCapClassRc}
	A.~E. Gamal and S.~Zahedi, ``Capacity of a class of relay channels with
	orthogonal components,'' \emph{IEEE Transactions on Information Theory},
	vol.~51, no.~5, pp. 1815--1817, May 2005.
	
	\bibitem{ref:zaidiRcSide}
	A.~Zaidi, L.~Vandendorpe, and P.~Duhamel, ``Lower bounds on the capacity
	regions of the relay channel and the cooperative relay-broadcast channel with
	non-causal side information,'' in \emph{IEEE International Conference on
		Communications, 2007, ICC'07}, June 2007, pp. 6005--6011.
	
	\bibitem{ref:kingdiissertation}
	R.~C. King, \emph{Multiple access channels with generalized feedback}.\hskip
	1em plus 0.5em minus 0.4em\relax Stanford, CA: Ph.D. dissertation, Stanford
	University, May 1978.
	
	\bibitem{ref:carleial}
	A.~B. Carleial, ``Multiple-access channels with different generalized feedback
	signals,'' \emph{IEEE Transaction on Information Theory}, vol. IT-28, pp.
	841--850, November 1982.
	
	\bibitem{ref:bookElgamal}
	A.~E. Gamal and Y.-H. Kim, \emph{Network information theory. Cambridge}.\hskip
	1em plus 0.5em minus 0.4em\relax {Cambridge Univ.} Press., Cambridge, U.K.,
	2011.
	
	\bibitem{ref:coverElgamalTh}
	T.~M. Cover and A.~E. Gamal, ``Capacity theorems for relay channel,''
	\emph{IEEE Transaction on Information Theory}, vol.~25, no.~5, pp. 572--584,
	September 1979.
	
	\bibitem{ref:Liu}
	R.~Liu, I.~Maric, P.~Spasojevic, and R.~Yates, ``Discrete memoryless
	interference and broadcast channels with confidential messages: Secrecy rate
	regions,'' \emph{IEEE Transactions on Information Theory}, vol.~54, no.~6,
	pp. 2493--2507, June 2008.
	
	\bibitem{ref:Xu}
	P.~Xu, Z.~Ding, X.~Dai, and K.~K. Leung, ``A general framework of wiretap
	channel with helping interference and state information,'' \emph{IEEE
		Transactions on Information Forensics and Security}, vol.~9, no.~2, pp.
	182--195, February 2014.
	
	\bibitem{ref:bookCover}
	T.~M. Cover and J.~A. Thomas, \emph{Elements of Information Theory},
	2nd~ed.\hskip 1em plus 0.5em minus 0.4em\relax {Wiley}, New York, 2006.
	
\end{thebibliography}

\end{document}